\documentclass[11pt]{article}

\usepackage{geometry}
\usepackage{amsmath,amsthm,amssymb,amsfonts}
\usepackage[T1]{fontenc}
\usepackage{lmodern}
\usepackage{microtype}
\usepackage{hyperref}
\usepackage{mathtools}
\usepackage{xspace}
\usepackage{xifthen}
\usepackage{xcolor}
\usepackage[noend]{algpseudocode}
\usepackage{algorithm}
\usepackage{tikz}
\usepackage{pgfplots}
\pgfplotsset{width=7.6cm,compat=1.9} 
\usetikzlibrary{arrows}
\usetikzlibrary{decorations.pathreplacing}

\theoremstyle{plain}
\newtheorem{theorem}{Theorem}

\newtheorem{lemma}{Lemma}
\newtheorem{corollary}{Corollary}
\newtheorem{definition}{Definition}

\newcommand{\pname}[0]{\textsc{BuildQuadTree}}

\title{Self-stabilizing Overlays for high-dimensional Monotonic Searchability\footnote{
This work is partially supported by the German Research Foundation (DFG) within the Collaborative Research Center ``On-The-Fly Computing'' (SFB 901).
}}
\author{
  Michael Feldmann\footnote{Paderborn University, Germany. \texttt{\{michael.feldmann, ckolb, scheideler\}@upb.de}}
  \and
  Christina Kolb\footnotemark[2]
  \and
  Christian Scheideler\footnotemark[2]
}

\begin{document}

\maketitle

\begin{abstract}
	We extend the concept of monotonic searchability~\cite{DBLP:conf/opodis/ScheidelerSS15}~\cite{DBLP:conf/wdag/ScheidelerSS16} for self-stabilizing systems from one to multiple dimensions.
	A system is self-stabilizing if it can recover to a legitimate state from any initial illegal state.
	These kind of systems are most often used in distributed applications.
	Monotonic searchability provides guarantees when searching for nodes while the recovery process is going on.
	More precisely, if a search request started at some node $u$ succeeds in reaching its destination $v$, then all future search requests from $u$ to $v$ succeed as well.
	Although there already exists a self-stabilizing protocol for a two-dimensional topology~\cite{DBLP:journals/tcs/JacobRSS12} and an universal approach for monotonic searchability~\cite{DBLP:conf/wdag/ScheidelerSS16}, it is not clear how both of these concepts fit together effectively.
	The latter concept even comes with some restrictive assumptions on messages, which is not the case for our protocol.
	We propose a simple novel protocol for a self-stabilizing two-dimensional quadtree that satisfies monotonic searchability.
	Our protocol can easily be extended to higher dimensions and offers routing in $\mathcal O(\log n)$ hops for any search request.
\end{abstract}

\section{Introduction}
Due to the growth and relevance of the Internet, the importance of distributed systems is increasing.
Such systems are needed, for instance, in social media networks or multiplayer games and have to support a large number of participants.
However, as soon as such a system has become large, the occurrence of changes or faults are not an exception but the rule.
In order to recover from an arbitrary state to a legitimate one, distributed protocols are needed that are \emph{self-stabilizing}.

Most of the proposed self-stabilizing protocols only show that the system \emph{eventually} converges to a legitimate state, without considering the \emph{monotonicity} of the actual recovery process.
Monotonicity means that the functionality of the system regarding a specific property never gets worse as time progresses, i.e., for two points in time $t, t'$ with $t < t'$, the functionality of the system is better in $t'$ than in $t$.

In this paper we are interested in \emph{searching}, as this is one of the most important operations in a distributed system.
We study systems that satisfy \emph{monotonic searchability}: If a search request for node $w$ starting at node $v$ succeeds at time $t$, then every search request for $w$ initiated by $v$ at time $t'>t$ succeeds as well.

Previous work on monotonic searchability~\cite{DBLP:conf/opodis/ScheidelerSS15}~\cite{DBLP:conf/wdag/ScheidelerSS16} proposed self-stabilizing protocols for one-dimensional topologies (for instance a sorted list).
Still, up to this point it is not known how to come up with an efficient self-stabilizing protocol for high-dimensional settings that satisfies monotonic searchability.
High-dimensional settings are relevant for example in wireless ad-hoc networks or social networks where processes are defined by multiple parameters.

This paper introduces a novel protocol \pname{} for a self-stabilizing quadtree along with a routing protocol \textsc{SearchQuad} that satisfies monotonic searchability and terminates after $\mathcal O(\log n)$ hops on any input.
To the best of our knowledge, this is the first protocol that combines self-stabilization and monotonic searchability for the two-dimensional case.
In addition, one can easily extend our protocols in order to work for multiple dimensions.
For the two-dimensional case, we expand the notion of monotonic searchability to an even stronger and more realistic property, which we call \emph{geographic monotonic searchability} and show that \textsc{SearchQuad} satisfies this property as well.
Our protocols stand out due to their simplicity and elegance and do not enforce restrictive assumptions on messages, as it has been done for the universal approach~\cite{DBLP:conf/wdag/ScheidelerSS16}.

\subsection{Model}
We consider a two-dimensional square $P$ of unit side length and model the distributed system as a directed graph $G=(V,E)$ with $n$ nodes. 
Each node $v \in V$ represents a single peer and can be identified via its \emph{unique position} in $P$ given by \emph{coordinates} $(v_x,v_y) \in [0,1]^2$.
We define $\vert \vert (u,v)\vert \vert$ as the Euclidean distance between two nodes $u,v \in V$, i.e., $\vert \vert (u,v)\vert \vert = \sqrt{(u_x - v_x)^2 + (u_y - v_y)^2}$.
Additionally, each node $v$ maintains local protocol-based variables and has a \textit{channel} $v.Ch$, which is a system-based variable that contains incoming messages. 
We assume a channel to be able to store any finite number of messages.
Messages are never duplicated or get lost in the channel.
If a node $u$ knows the coordinates of some other node $v$, then $u$ can send a message $m$ to $v$ by putting $m$ into $v.Ch$.
There is a directed edge $(u,v) \in E$ whenever $u$ stores $(v_x,v_y)$ in its local memory or when there is a message in $u.Ch$ carrying $(v_x,v_y)$. In the former case, we call that edge \textit{explicit} and in the latter case we call that edge \textit{implicit}.

Nodes may execute \textit{actions}: An action is a standard procedure and has the form $\langle label \rangle(\langle parameters \rangle):\langle command \rangle$, where $label$ is the name of that action, $parameters$ defines the set of parameters and $command$ defines the statements that are executed when calling that action.
It may be called locally or remotely, i.e., every message that is sent to a node has the form $\langle label \rangle(\langle parameters \rangle)$.
An action in a process $v$ is \emph{enabled} if there is a request for calling it in $v.Ch$. 
Once the request is processed, it is removed from $v.Ch$.
There is a special action called \textsc{Timeout} that is not triggered via messages but is executed periodically by each node.

We define the \textit{system state} to be an assignment of a value to every node's variables and messages to each channel.
A \textit{computation} is an infinite sequence of system states, where the state $s_{i+1}$ can be reached from its previous state $s_i$ by executing an action in $s_i$.
We call the first state of a given computation the \textit{initial state}.
We assume \textit{fair message receipt}, meaning that every message of the form $\langle label \rangle(\langle parameters \rangle)$ that is contained in some channel, is eventually processed.
We place no bounds on message propagation delay or relative node execution speed, i.e.,  we allow fully asynchronous computations and non-FIFO message delivery.
Our protocol does not manipulate node coordinates and thus only operates on them in \emph{compare-store-send} mode, i.e., we are only allowed to compare node coordinates to each other, store them in a node's local memory or send them in a message.

We assume for simplicity that there are no \emph{corrupted coordinates} in the initial state of the system, i.e., coordinates of unavailable nodes.
One could use failure detectors to solve this, but this is not within the scope of this paper, since without them the problem of guaranteeing monotonic searchability is still non-trivial.
Having node coordinates to be read-only also makes sense in our setting, as these are usually delivered by an external component that is not in control of our protocol, for instance like GPS.
In initial states there may exist \emph{corrupted messages} in node channels, i.e., messages containing false information.
We will argue that at a certain point in time, all of these messages will be processed and no more corrupted messages are in the system.

Nodes are able to issue search requests at any point in time: A search request is a message \Call{Search}{$v$, $(x,y)$}, where $v$ is the sender of the message and $(x,y) \in [0,1]^2$ are the coordinates we want to search for.
A search request is delegated along edges in $G$ according to a given routing protocol, until the request \emph{terminates}, i.e., either the node with coordinates $(x,y)$ is reached or the request cannot be forwarded anymore.
Note that $(x,y)$ do not necessarily need to be coordinates of an existing node, i.e., in such a case the routing protocol may just stop at some node that is close to $(x,y)$.
Upon termination at node $w$, the reference of $w$ is returned to the sender $v$ (in the pseudocode we indicate this via a return statement).

\subsection{Problem Statement} \label{sec:problem_statement}
In this paper we consider the standard definition for self-stabilization:

\begin{definition}[Self-stabilization] \label{def:self_stabilization}
A protocol is \emph{self-stabilizing} w.r.t. a set of legitimate states if it satisfies the following two properties:
\begin{enumerate}
	\item Convergence: Starting from an arbitrary system state, the protocol is guaranteed to arrive at a legitimate state.
	\item Closure: Starting from a legitimate state, the protocol remains in legitimate states thereafter.
\end{enumerate}
\end{definition}

We are interested in \emph{topological self-stabilization} in this paper, meaning that our self-stabilizing protocol is allowed to perform changes to the overlay network $G$.
In order for our protocol to work, we require the directed graph $G$ containing all explicit and implicit edges to be at least weakly connected initially.
Once there are multiple weakly connected components in $G$, these components cannot be connected to each other anymore as it has been shown in~\cite{DBLP:journals/tcs/NorNS13} for compare-store-send protocols.
For a graph that contains multiple weakly connected components, our protocol converts each of these components to our desired topology.

Consider the following definition of (standard) monotonic searchability:

\begin{definition}[Monotonic Searchability]\label{def:monotonic_searchability:classical}
	A self-stabilizing protocol satisfies \emph{monotonic searchability} according to some routing protocol $\mathcal R$ if it holds for any pair of nodes $v,w$ that once a search request $\textsc{Search}(v, (w_x,w_y))$ returns $w$ at time $t$, any search request $\textsc{Search}(v, (w_x,w_y))$ initiated at at time $t' > t$ also returns $w$.
\end{definition}

Realizing monotonic searchability in self-stabilizing systems is a non-trivial problem, because once a $\textsc{Search}(v, (w_x,w_y))$ request returns $w$ to $v$, it cannot trivially be guaranteed that $w$ is found again by $v$ at later stages, due to the modification of edges by the self-stabilizing protocol.

The above definition differs in a minor detail compared to the definition stated in~\cite{DBLP:conf/opodis/ScheidelerSS15}~\cite{DBLP:conf/wdag/ScheidelerSS16}: The initial search request issued by $v$ terminates at time $t$, but Scheideler et. al. define the time step $t$ to be the one at which the initial search request was generated by $v$.
They use a probing approach to check for a node $v$ whether $v$ is still waiting for the result of a previously issued search request and cache all search requests searching for the same target.
The same approach can be applied to our protocol as well to overcome this, but for the sake of simplicitiy we use the slightly modified definition stated above.

In two-dimensional scenarios it is more realistic to search for geographic positions rather than for concrete node addresses. 
To handle this, we introduce the following definition of geographic monotonic searchability. 

\begin{definition}[Geographic Monotonic Searchability]\label{def:monotonic_searchability:quad}
	Let $(x,y) \in [0,1]^2$ be an arbitrary position in $P$.
	Let $w \in V$ be the node that would be returned by $\textsc{Search}(v, (x,y))$ if the system is in a legitimate state.
	A self-stabilizing protocol satisfies \emph{geographic monotonic searchability} according to some routing protocol $\mathcal R$ if in case the system is in an arbitrary state and $\textsc{Search}(v, (x,y))$ returns $w$ at time $t$, then any request $\textsc{Search}(v, (x,y))$ initiated at time $t' > t$ also returns $w$.
\end{definition}

This definition is even stronger than (standard) monotonic searchability, i.e., a protocol satisfying geographic monotonic searchability also satisfies monotonic searchability.
Therefore we focus on geographic monotonic searchability for the rest of this paper.

We aim to solve the following problem: Given a weakly connected graph of $n$ nodes with coordinates in $P$, construct a self-stabilizing protocol along with a routing protocol such that geographic monotonic searchability is satisfied.

\subsection{Our Contribution}
In the following we summarize our contributions:

\begin{itemize}
	\item[(1)] We propose a novel self-stabilizing protocol \pname{} that arranges the nodes in a quadtree.
	\pname{} is based on a special kind of subdivision of $P$ into subareas inducing an ordering via a space-filling curve (see Section~\ref{sec:topology}) and the \textsc{BuildList} protocol (Section~\ref{sec:build_quad_tree:list}).
	To the best of our knowledge this is the first self-stabilizing protocol for the quadtree structure.
	\item[(2)] Along with the self-stabilizing protocol \pname{} we propose the routing protocol $\textsc{SearchQuad}$.
	When searching for coordinates $(x,y)$, the protocol returns the node $w$, which lies within the same subarea as $(x,y)$. 
	We show that \pname{} along with $\textsc{SearchQuad}$ satisfies geographic monotonic searchability (and thus also standard monotonic searchability).
	\item[(3)] We get an upper bound of $\mathcal O(\log n)$ on the number of hops for a search message (i.e, the amount of times a search message is delegated until it terminates) if we assume that the Euclidean distance $\vert\vert (u,v)\vert\vert$ between any pair of nodes $(u,v)\in V$ is at least $1/n$.
	This is particularly an improvement on the protocols proposed in~\cite{DBLP:conf/opodis/ScheidelerSS15}~\cite{DBLP:conf/wdag/ScheidelerSS16} regarding the maximum number of hops for searching a target, even for target addresses that do not exist (see Section~\ref{sec:related_work} on related work).
	\item[(4)] Finally, one can easily extend \pname{} and $\textsc{SearchQuad}$ to work in high-dimensional settings, realizing the first self-stabilizing protocol for octtrees - the high-dimensional equivalent of quadtrees - that even satisfies geographic monotonic searchability.
	This makes our protocols highly versatile.
\end{itemize}

The rest of the paper is structured as follows: After stating some related work, we describe our topology for the quadtree in Section~\ref{sec:topology}.
Then we present our novel protocol \pname{} in Section~\ref{sec:build_quad_tree} along with the routing protocol \textsc{SearchQuad}.
Before we conclude in Section~\ref{sec:conclusion}, we analyze our protocols in Section~\ref{sec:analysis} and discuss how to extend \pname{} and $\textsc{SearchQuad}$ to work in high-dimensional settings in Section~\ref{appendix:multidim}.

\subsection{Related Work} \label{sec:related_work}
Quadtrees have first been introduced in 1974 by R. A. Finkel and J.L. Bentley~\cite{DBLP:journals/acta/FinkelB74}.
Since then quadtrees and octrees are most often used in computational geometry~(for surveys consider for example~\cite{DBLP:reference/crc/Aluru04}~\cite{DBLP:conf/ssd/Samet89}).
There are also peer-to-peer approaches relying on quadtrees~\cite{DBLP:conf/ipsn/GaoGHZ04}~\cite{DBLP:journals/vldb/TaninHS07}.
Still, the problem of designing a self-stabilizing protocol that arranges peers in a quadtree is untouched until today.

The concept of self-stabilization has first been introduced by E. W. Dijkstra in 1974 via a self-stabilizing token-based ring~\cite{DBLP:journals/cacm/Dijkstra74}.
This led to the introduction of various other self-stabilizing protocols for network topologies such as sorted lists~\cite{DBLP:conf/alenex/OnusRS07}~\cite{DBLP:journals/mst/GallJRSST14}, De Bruijn graphs~\cite{DBLP:conf/sss/RichaSS11}, Chord graphs~\cite{DBLP:journals/mst/KniesburgesKS14}, Skip graphs~\cite{DBLP:journals/tcs/ClouserNS12}~\cite{DBLP:conf/podc/JacobRSST09} and many more.
A universal approach that is able to derive self-stabilizing protocols for several types of topologies was introduced in~\cite{DBLP:journals/tcs/BernsGP13}.
Interestingly, topological self-stabilization in two- or high-dimensional settings is barely investigated until now: There exists only a single self-stabilizing protocol that transforms any weakly connected graph into a two-dimensional topology - the Delaunay graph~\cite{DBLP:journals/tcs/JacobRSS12}.
Unfortunately, it seems non-trivial to extend this such that monotonic searchability is satisfied, without resorting to expensive mechanisms like broadcasting.
Also, one cannot guarantee searching in $\mathcal O(\log n)$ hops in the Delaunay graph, as its diameter is too large.

Research on monotonic searchability was initiated in~\cite{DBLP:conf/opodis/ScheidelerSS15}, where the authors presented a self-stabilizing protocol for the sorted list that satisfies monotonic searchability.
They also showed that providing monotonic searchability is impossible in general when the system contains corrupted messages.
However, this property is restricted to cases where the desired topology to which the graph should converge is clearly defined, forcing the underlying protocol to eventually remove an explicit edge if it is not part of the desired topology.
This is not the case for our topology, because once a specific explicit edge (which we define as \emph{quad edge} later on) is generated by our protocol it is never deleted, so the legitimate state $s$ that we reach is dependent on the specific computation done before reaching $s$.
Therefore we do not need to enforce any restrictions on messages, as routing is done via quad edges only.
Building on that research, the same authors presented a universal approach for maintaining monotonic searchability at DISC 2016 along with a generic routing protocol that can be applied to a wide range of topologies~\cite{DBLP:conf/wdag/ScheidelerSS16}.
However, adapting their protocol to specific topologies comes at the cost of convergence times and additional message overhead.
Furthermore, search request forwarded via their generic routing protocol might travel $\Omega(n)$ hops when searching for non-existing nodes, whereas our routing protocol only needs $\mathcal O(\log n)$ hops on any input to terminate, while still satisfying monotonic searchability.
In addition to this, our protocol \pname{} is simpler and also more lightweight regarding the message overhead.
This is mostly due to the simplicity of the quadtree topology.

Closest but different from our notion of monotonic searchability is the notion of \emph{monotonic stabilization}~\cite{DBLP:conf/opodis/YamauchiT10}.
A self-stabilizing protocol is monotonically stabilizing, if every change done by its nodes is making the system approach a legitimate state and if every node changes its output only once.
The authors show that processes have to exchange additional information in order to satisfy monotonic stabilization.

For the computation of an ordering, we use a space-filling curve similar to the Morton-curve~\cite{morton1966computer}, as it matches the structure of the quadtree best.
Other curves like the Hilbert-curve would also work in principle, however, using them would make the presentation of our ideas way more harder.

\section{Topology and Legitimate State} \label{sec:topology}
In this section we introduce our desired topology for the quadtree and define what it means for our system to be in a legitimate state.
We first provide some intuition: Given a set $V$ of $n$ nodes with coordinates in $P$, we first cut the area $P$ into two equally sized \emph{subareas}, via a vertical cut.
This is done recursively for each subarea, alternating between vertical and horizontal cuts, as long as the subarea contains more than one node.
Once this is done, we can define a total order on all nodes in $P$, that is used to connect the nodes into a (doubly-linked) sorted list.
Based on this list and the generated subareas, we establish further edges, which we use for the routing protocol.

More formally, let us consider the recursive algorithm \textsc{QuadDivision} (see Algorithm~\ref{algo:quad_division} for pseudocode) having a set of nodes, a (sub-)area and a flag indicating the next cut (vertical or horizontal) as input.

\begin{algorithm}[ht]
\caption{Quad Division Algorithm}
\label{algo:quad_division}
\begin{algorithmic}[1]
\Procedure{QuadDivision}{$V$, $P$, $cut$}
	\If{$cut = 1$}
		\State Perform vertical cut on $P$, resulting in $P = P_1 \cup P_2$
	\Else
		\State Perform horizontal cut on $P$, resulting in $P = P_1 \cup P_2$
	\EndIf
	\State $S \gets \emptyset$
	\If{$P_1$ contains at most one node out of $V$}
		\State $S \gets S \cup \{P_1\}$
	\Else
		\State $S \gets S\ \cup$ \Call{QuadDivision}{$\{v \in V\ |\ v \in P_1\}$, $P_1$, $|cut-1|$}
	\EndIf
	\If{$P_2$ contains at most one node out of $V$}
		\State $S \gets S \cup \{P_2\}$
	\Else
		\State $S \gets S\ \cup$ \Call{QuadDivision}{$\{v \in V\ |\ v \in P_2\}$, $P_2$, $|cut-1|$}
	\EndIf
	\State \Return $S$
\EndProcedure
\end{algorithmic}
\end{algorithm}

Initially we call \Call{QuadDivision}{$V$, $P$, $1$} and thus perform a vertical cut on $P$, dividing it into equally sized subareas $P_1$ and $P_2$.
Then we call \textsc{QuadDivision} recursively on $P_1$ and $P_2$ as long as they contain more than one node.
We say a subarea $A$ \emph{contains} node $v$ (or conversely, node $v$ is contained in the subarea $A$), denoted by $v \in A$, if $v$'s coordinates $(v_x, v_y)$ lie within $A$.
If a subarea $A$ contains no node from $V$, we say that $A$ is \emph{empty}.
For simplicity, we assume that nodes do not lie on the boundaries of subareas, as this would disturb the presentation of our algorithm, but the problem can easily be resolved in practice.
\Call{QuadDivision}{$V$, $P$, $1$} returns the set $S$ of subareas that contain at most one node.
Figure~\ref{fig:cutting_example} shows an example for a sequence of cuts with $4$ nodes $v_1, \ldots, v_4$.
Note that upon termination, \textsc{QuadDivision} returns $5$ subareas (one subarea for each node $v_i$ and the empty subarea on the bottom left).

\begin{figure}[ht]
 	\centering
 	\includegraphics[width=\textwidth]{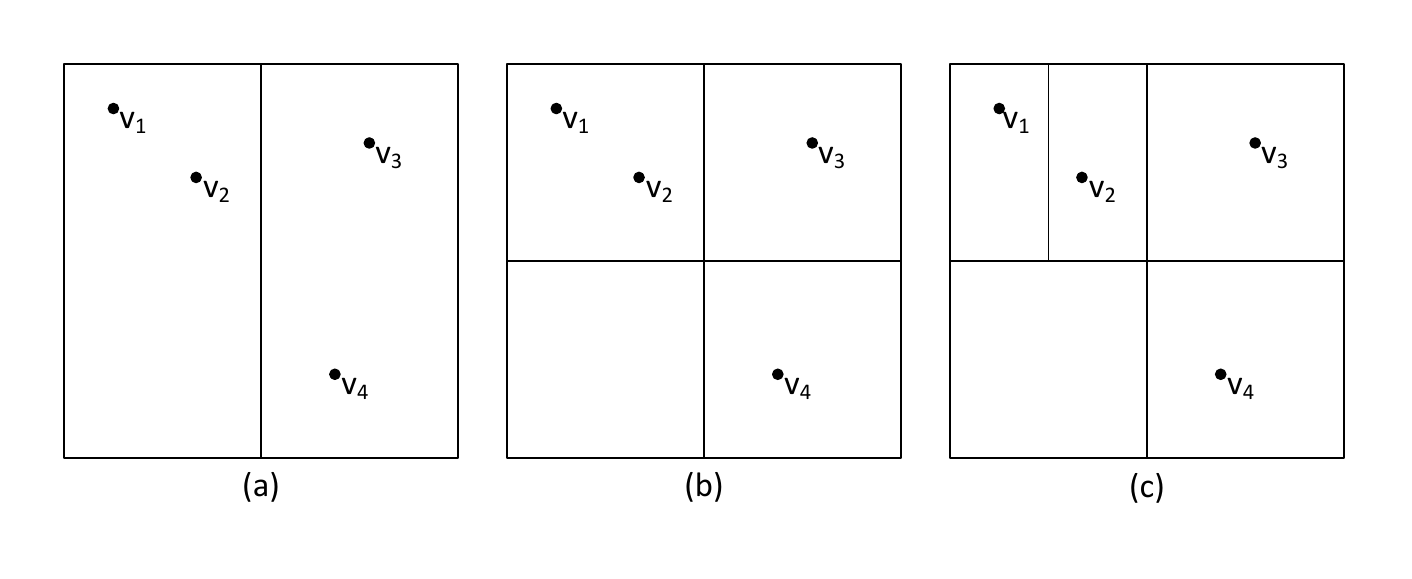}
 	\caption{Illustration of \textsc{QuadDivision} performed on nodes $v_1,\ldots,v_4$. (a) illustrates the first vertical cut on $P$. (b) illustrates the horizontal cuts done to subareas $P_1$ and $P_2$. (c) illustrates the final vertical cut before termination.}
 	\label{fig:cutting_example}
\end{figure}

In the following we want to view the output of \textsc{QuadDivision} as a binary tree $T$: The root node corresponds to the whole square $P$.
An inner node of $T$ corresponding to a (sub-)area $P$ has two child nodes: Cutting $P$ into two subareas $P_1$ and $P_2$, the \emph{left child} represents the subarea that lies west of the other (when performing a vertical cut on $P$) or north of the other (when performing a horizontal cut on $P$).
Similarly, the \emph{right child} represents the subarea that lies east of the other (when performing a vertical cut on $P$) or south of the other (when performing a horizontal cut on $P$).
The binary tree is the unique minimal such tree having no leaf node $t \in T$ correspond to a subarea of $P$ that contains more than one node $v \in V$.
Note that this makes nodes $v \in V$ correspond to leaf nodes in $T$, but a leaf node $t \in T$ does not necessarily correspond to a node in $V$, as the subarea represented by $t$ may be empty.
Figure~\ref{fig:tree_example} shows the corresponding binary tree $T$ to the previous example from Figure~\ref{fig:cutting_example}.

\begin{figure}[ht]
 	\centering
 	\includegraphics[width=\textwidth]{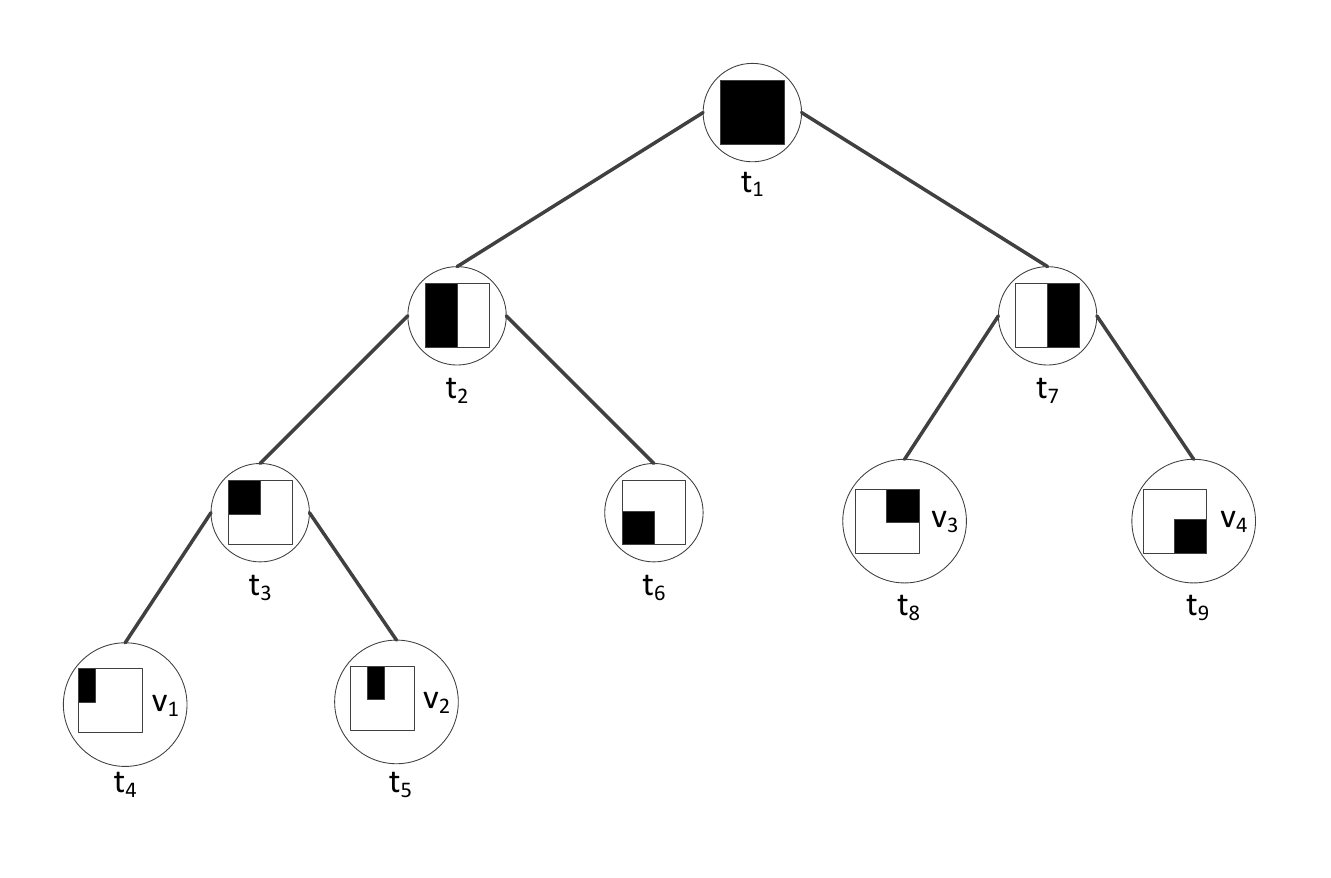}
 	\caption{Corresponding binary tree to the previous example from Figure~\ref{fig:cutting_example}. The subareas marked in black are the subareas that are represented by the corresponding tree node.
 	Performing a depth-first search on the tree, when always going to the left child first, yields the total order $v_1 \prec v_2 \prec v_3 \prec v_4$.}
 	\label{fig:tree_example}
\end{figure}

Using the binary tree notation, we can define a total order on $V$:

\begin{definition}[Two-Dimensional Ordering] \label{def:total_order}
	Let $T$ the be tree corresponding to the subareas that are returned by \textsc{QuadDivision}$(V$, $P$, $1)$.
	The total order $\prec$ is given by the depth-first search (DFS) traversal of $T$, always going to the left child first.
\end{definition}

When comparing nodes $v$ and $w$ via $\prec$ we say that $v$ is \emph{left} of $w$, if $v \prec w$, otherwise $v$ is \emph{right} of $w$ (note that either of the two cases always holds as we assume node coordinates to be unique).
In addition we say that $v$ is $w$'s \emph{closest left neighbor} if $v \prec w$ and there is no node $u$ with $v \prec u \prec w$.
Analogously we define a node $v$ being the \emph{closest right neighbor} of $w$.

As nodes in the binary tree $T$ correspond to subareas of $P$ and vice versa, we use them interchangeably for the rest of the paper.
We say that a node $t \in T$ \emph{represents} a subarea $A$, if $A$ is the corresponding subarea to $t$.
The next definition introduces important notation in order to define the legitimate state:

\begin{definition} \label{def:target_subareas}
	Let $T$ be the tree representing the subareas that are returned by a \textsc{Quad-Division}$(V$, $P$, $1)$ call.
	For a node $v \in V$, denote the leaf node representing the subarea that contains $v$ by $A(v)$.
	Define the set $Q(v)$ as the set of subareas represented by nodes $t \in T$ such that the following holds: 
	\begin{itemize}
		\item[(a)] If $t \in Q(v)$, then the subarea represented by $t$ does not contain $v$.
		\item[(b)] If $t \in Q(v)$, then the subarea represented by the parent node of $t$ contains $v$.
		\item[(c)] Combining all subareas in $Q(v)$ with $A(v)$ yields the whole square $P$.
	\end{itemize} 
\end{definition}

As an example consider again Figure~\ref{fig:tree_example}: The set $Q(v_1)$ consists of the subareas $t_5, t_6$ and $t_7$, as the combination of these with the subarea $t_4$ containing $v_1$ yield the square $P$.
Note that for instance $t_8 \in Q(v_1)$ would violate condition (b).

Using the total order $\prec$ we are now ready to define the legitimate state of our system, i.e., the topology that should be reached by our self-stabilizing protocol:

\begin{definition}[Legitimate State] \label{def:legal_state}
	The system is in a legitimate state if the graph induced by the explicit edges satisfies the following conditions:
	\begin{itemize}
		\item[(a)] Each node $v$ is connected to its closest left and right neighbor w.r.t. $\prec$.
		\item[(b)] For each non-empty subarea $A \in Q(v)$, $v$ is connected to exactly one node $w \in A$.
	\end{itemize}
\end{definition}

Consider Figure~\ref{fig:legal_state_example} showing a possible legitimate state for the nodes from Figure~\ref{fig:cutting_example}.

\begin{figure}[ht]
 	\centering
 	\includegraphics[scale=1]{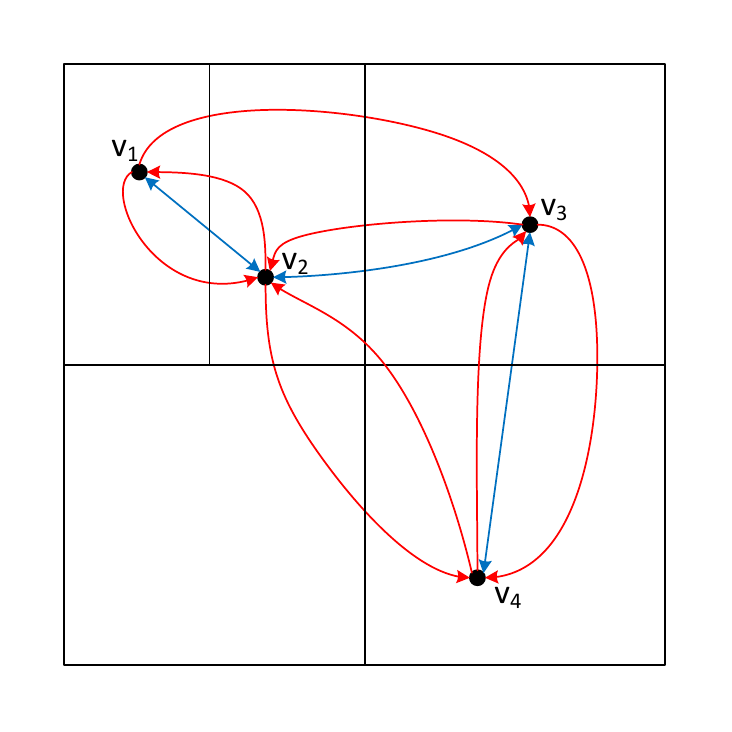}
 	\caption{A possible legitimate state for the system from Figure~\ref{fig:cutting_example}. List edges are indicated in blue, quad edges in red.}
 	\label{fig:legal_state_example}
\end{figure}

Note that we do not clearly define nodes for $v$ to connect to in condition (b) more specifically, we just want to make sure that $v$ is able to reach the subarea directly via an outgoing edge in case the subarea contains nodes.
As it turns out, this helps us in order to achieve geometric monotonic searchability.
We want to emphasize that edges in $T$ are not part of the legitimate state, as we use the binary tree to illustrate our approach and only let nodes compute necessary parts of the tree locally.

\section{Protocol Description} \label{sec:build_quad_tree}
In this section we describe the self-stabilizing \pname{} protocol and the routing algorithm \textsc{SearchQuad}.
We first define the protocol-based variables for each node.
We denote by $\perp$ that the variable does not contain any node.
Each node $v \in V$ maintains the following variables:

\begin{itemize}
	\item Variables $v.left, v.right \in V \cup \{\perp\}$ storing $v$'s left and right neighbor.
	\item A set $v.Q \subset V$ storing a single node $w \in V$ for each non-empty subarea $A \in Q(v)$ such that $w \in A$.
\end{itemize}

We refer to the edges represented by variables $v.left$ and $v.right$ as \emph{list edges} and to edges $(v,w)$ with $w \in v.Q$ as \emph{quad edges}.
Observe that a node $w$ is allowed to be contained in both $v.left$ (resp. $v.right$) and $v.Q$ simultaneously in a legitimate state.
The reason for this is that we allow the delegation of search messages only via quad edges (as we will see in Section~\ref{sec:protocol:routing}), so if $v$ wants to delegate a search message to the subarea containing one of its list edges, it has to make sure that there is a node in $v.Q$ for this area.

Before we can describe how we establish the correct list and quad edges, we shortly describe how a node $v$ that knows some node $w$ is able to locally determine whether $v \prec w$ or $w \prec v$ holds: $v$ just calls \Call{QuadDivision}{$\{v,w\}$, $P$, $1$} and gets a binary tree with subareas containing $v$ and $w$ as leaf nodes.
Performing a DFS on that tree as described earlier yields either $v \prec w$ or $w \prec v$.

It is important to note that using the same approach, $v$ is also able to compute the set $Q(v)$ for the current system state: $v$ just calls \Call{QuadDivision}{$\{v,v.left, v.right\}$, $P$, $1$}.
It is easy to see that the corresponding tree contains all nodes representing subareas in $Q(v)$, so $v$ just has to check each node in the tree for the properties from Definition~\ref{def:target_subareas}.
Obviously, as long as $v.left$ and $v.right$ are still subject to changes, $Q(v)$ also changes, but we will show later that by the way we defined our protocol, $Q(v)$ monotonically increases, s.t. none of the proposed properties are violated.

We now describe how we build the correct list edges at each node and then proceed with the description for quad edges.
As we have to perform actions in both parts periodically, we split the \textsc{Timeout} action into subroutines \textsc{ListTimeout} and \textsc{QuadTimeout}.
For list edges, we extend the \textsc{BuildList} protocol that is based on~\cite{DBLP:conf/alenex/OnusRS07} for the one-dimensional case to the two-dimensional case.

\subsection{List Edges} \label{sec:build_quad_tree:list}
The base of our self-stabilizing protocol consists of a sorted list for all nodes $v \in V$ based on the ordering $\prec$ from Definition~\ref{def:total_order} (see Algorithm~\ref{algo:build_list} for pseudocode).

\begin{algorithm}[ht]
\caption{The \textsc{BuildList} Protocol (executed at node $v$)}
\label{algo:build_list}
\begin{algorithmic}[1]
\Procedure{ListTimeout}{}
	\State Consistency check for $v.left$ and $v.right$ w.r.t. $\prec$
	\State $v.left \gets$ \Call{Linearize}{$v$} \Comment{Send \Call{Linearize}{$v$} message to $v.left$}
	\State $v.right \gets$ \Call{Linearize}{$v$}
\EndProcedure
\State
\Procedure{Linearize}{$w$}
	\State Consistency check for $v.left$ and $v.right$ w.r.t. $\prec$
	\If{$w \prec v.left$}
		\State $v.left \gets$ \Call{Linearize}{$w$}
	\EndIf
	\If{$v.left \prec w \prec v$}
		\State $w \gets$ \Call{Linearize}{$v.left$}
		\State $v.left \gets w$
	\EndIf
	\If{$v \prec w \prec v.right$}
		\State $w \gets$ \Call{Linearize}{$v.right$}
		\State $v.right \gets w$
	\EndIf
	\If{$w \succ v.right$}
		\State $v.right \gets$ \Call{Linearize}{$w$}
	\EndIf
\EndProcedure
\end{algorithmic}
\end{algorithm}

The main idea of \textsc{BuildList} is that each node of $V$ keeps its closest left and right neighbor in $v.left$ and $v.right$. 
More concretely, the protocol consists of two actions called \textsc{ListTimeout} and \textsc{Linearize}. 
\textsc{ListTimeout} is periodically executed and \textsc{Linearize} can be called locally or remotely. 

Whenever \textsc{ListTimeout} or \textsc{Linearize} is executed, $v$ first performs a local consistency check on its variables $v.left$ and $v.right$: It might happen that in initial states $v.left \succ v$ (or $v.right \prec v$). 
If that is the case, $v$ sets $v.left$ (or $v.right$) to $\perp$ and locally calls \Call{Linearize}{$w$} for the removed value $w$. 
In addition to the above described consistency check, $v$ \emph{introduces} itself to $v.left$ and $v.right$ in \textsc{ListTimeout} by sending a \Call{Linearize}{$v$} message to them.

In case $v$ processes a \Call{Linearize}{$w$} request, $v$ does the following: $v$ sets $v.left=w$, if $w$ is left of $v$ and closer to $v$ than $v.left$, i.e., $v.left \prec w \prec v$.
In that case, $v$ delegates the value $w'$ that got replaced by $w$ in $v.left$ to $w$, i.e., $v$ calls \Call{Linearize}{$w'$} on $w$.
In case $v.left=\perp$, $v$ just sets $v.left=w$. 
In case $w$ is right of $v$, $v$ proceeds analogously for $v.right$. 

Note that node references are never deleted but delegated until the referenced node arrives at the correct spot in the sorted list. From~\cite{DBLP:conf/alenex/OnusRS07} we derive the following result.
The proof works the same as for the one-dimensional setting, we just replace the (one-dimensional) operator $<$ by $\prec$.

\begin{lemma} \label{lemma:build_list:self_stabilization}
	\textsc{BuildList} is self-stabilizing.
\end{lemma}

\subsection{Quad Edges} \label{sec:build_quad_tree:quad}
Now we describe the approach for generating quad edges.
Note that $v$ can easily check whether there exists a subarea $A \in Q(v)$ for which $v$ does not yet have a quad edge, by assigning each $w \in v.Q$ to the subarea in $Q(v)$ that contains $w$.

The protocol consists of actions \textsc{QuadTimeout} and \textsc{QLinearize} (see Algorithm~\ref{algo:build_quad_tree}).
Before executing any statement of any of these actions, a node $v$ always checks its set $v.Q$ for consistency, ensuring that no two nodes $w_1, w_2 \in v.Q$ are contained in the same subarea $A \in Q(v)$.
In case $v$ finds out that $w_1, \ldots , w_k \in v.Q$ are contained in the same subarea $A \in Q(v)$ (which may happen in an initial state), we only keep one of these nodes (arbitrarily chosen) and delegate all other nodes $w_i$ to \textsc{BuildList} by calling \Call{Linearize}{$w_i$}.

In \textsc{QuadTimeout}, $v$ chooses a node $w$ from its set $v.Q$ in round-robin fashion and delegates $w$ to \textsc{BuildList}.
This has to be done to ensure that the sorted list converges even if the initial weakly connected graph consists of quad edges only.
Afterwards $v$ introduces itself to its left and right neighbors $v.left$ and $v.right$ by calling \textsc{QLinearize} on them.
As part of the same \textsc{QLinearize} request, $v$ asks these nodes if they know a node $w \in A$, where $A \in Q(v)$ is a subarea, for which $v$ does not have a quad edge yet.
If that is the case, then $v$ will receive a \textsc{QLinearize} call containing the desired node $w$ as the answer.
The subarea $A$ is chosen in round-robin fashion as well, such that each subarea for which $v$ does not have a quad edge yet is chosen by $v$ eventually.
The reason for choosing nodes and subareas in round-robin fashion is that we do not want to overload the network with too many stabilization messages that are generated periodically.

Processing a \Call{QLinearize}{$w$, $A$} request at node $v$ works as follows: We delegate $w$ to \textsc{BuildList} and then check if $w$ is contained in a subarea $A' \in Q(v)$  for which there does not exist a node $w' \in v.Q$ with $w' \in A'$.
If that is the case, then $v$ does not have a quad edge to the subarea $A'$ yet, so $v$ includes $w$ into $v.Q$, which corresponds to $v$ generating a new quad edge $(v,w)$.
Finally $v$ generates an answer to $w$ as already described above, in case $v$ knows a node (including itself) that is contained in $A$.

\begin{algorithm}[ht]
\caption{Protocol for establishing quad edges (executed at node $v$)}
\label{algo:build_quad_tree}
\begin{algorithmic}[1]
\Procedure{QuadTimeout}{}
	\State Consistency check for $v.Q$
	\State Choose $w \in v.Q$ in round-robin fashion and call \Call{Linearize}{$w$} \label{algo:build_quad_tree:timeout:delegate}
	\State Determine $A(v)$ and $Q(v)$ via \textsc{QuadDivision}
	\State Choose $A \in Q(v)$ in round-robin fashion s.t. $\forall w \in v.Q: w \not \in A$ \label{algo:build_quad_tree:timeout:choose_subarea}
	\State $v.left \gets$ \Call{QLinearize}{$v$, $A$} \Comment{$A = \perp$ if no such $A$ exists}
	\State $v.right \gets$ \Call{QLinearize}{$v$, $A$}
\EndProcedure
\State
\Procedure{QLinearize}{$w$, $A$}
	\State Consistency check for $v.Q$
	\State \Call{Linearize}{$w$} \Comment{Delegation to \textsc{BuildList}}
	\State Determine $A(v)$ and $Q(v)$ via \textsc{QuadDivision}
	\If{$\exists A' \in Q(v)\ \forall w' \in v.Q: w' \not \in A'$}
		\State $v.Q \gets v.Q \cup w$ \Comment{Generates new quad edge $(v,w)$}
	\EndIf
	\If{$A \neq \perp \wedge\ \exists w' \in v.Q \cup v: w' \in A$}
		\State $w \gets$ \Call{QLinearize}{$w'$, $\perp$} \Comment{Answers $w$ so $w$ can generate quad edge $(w,w')$}
	\EndIf
\EndProcedure
\end{algorithmic}
\end{algorithm}

\subsection{Routing} \label{sec:protocol:routing}
As the last part of this section, we state the routing protocol \textsc{SearchQuad} for our topology (see Algorithm~\ref{algo:search_quad_tree} for pseudocode).

Before a node $v$ processes a search message, it first performs the same consistency checks on its set $v.Q$ as it has been described previously.
This makes sure that our routing protocol is well-defined.
Now assume node $v$ wants to process a \Call{Search}{$u$, $(x,y)$} message.
Consider the subarea $A(v)$ and the set $Q(v)$ of subareas as defined in Definition~\ref{def:target_subareas}.
Then $v$ determines the subarea $A(x,y)$ out of $Q(v) \cup A(v)$ that contains the position $(x,y)$.
If $A(x,y) = A(v)$, then the algorithm terminates and returns $v$ itself to $u$ as the result.
Otherwise, $v$ delegates the \Call{Search}{$u$, $(x,y)$} message to the node $w \in v.Q$ with $w \in A(x,y)$.
If no edge to a node in $A(x,y)$ exists in $v.Q$, then the algorithm terminates and returns $v$ itself to $u$ as the result.
Consider Figure~\ref{fig:routing_example} for some examples.

\begin{figure}[ht]
 	\centering
 	\includegraphics[scale=1]{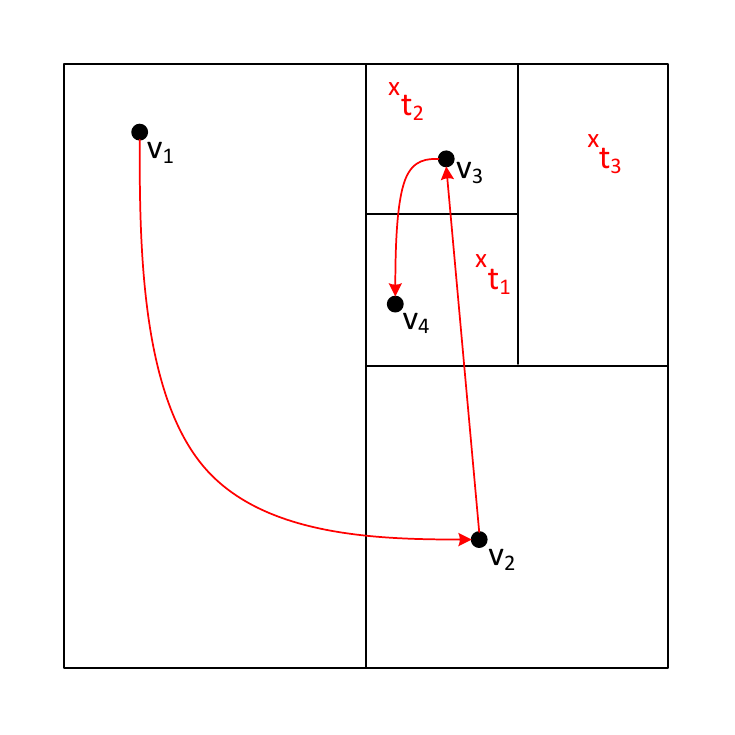}
 	\caption{Illustration for the delegation of different \textsc{Search} messages for target coordinates $t_1$, $t_2$ and $t_3$ starting at $v_1$. \textsc{Search}($v_1$,$t_1$) and \textsc{Search}($v_1$,$t_2$) returns the nodes that share the same subarea with the target point (traversing paths $(v_1,v_2,v_3,v_4)$ for $t_1$ and $(v_1,v_2,v_3)$ for $t_2$). The search for $t_3$ yields the path $(v_1,v_2,v_3)$ until \textsc{SearchQuad} terminates, as $v_3$ does not have a quad edge to the subarea containing $t_3$.}
 	\label{fig:routing_example}
\end{figure}

\begin{algorithm}[ht]
\caption{The \textsc{SearchQuad} Protocol (executed at node $v$)}
\label{algo:search_quad_tree}
\begin{algorithmic}[1]
\Procedure{Search}{$u, (x,y)$}
	\State Consistency check for $v.Q$
	\State Determine $A(v)$ and $Q(v)$ via \textsc{QuadDivision}
	\If{$(x,y) \in A(v)$}
		\State \Return $v$ \Comment{Search terminated - $v$ is returned to $u$ as the result}
	\Else
		\State Let $A(x,y) \in Q(v)$ with $(x,y) \in A(x,y)$
		\If{$\exists w \in v.Q: w \in A(x,y)$}
			\State $w \gets$ \Call{Search}{$u, (x,y)$} \Comment{Delegate request via quad edge $(v,w)$}
		\Else
			\State \Return $v$
		\EndIf
	\EndIf
\EndProcedure
\end{algorithmic}
\end{algorithm}

\section{Analysis} \label{sec:analysis}
\subsection{Self-stabilization Analysis} \label{sec:analysis:stabilization}
This section is dedicated to show that \pname{} is self-stabilizing according to Definition~\ref{def:self_stabilization}, i.e., \pname{} satisfies convergence and closure.

Recall that our system initially is given by an arbitrary weakly connected graph $G=(V,E)$.
As the graph may consist of both list- and quad edges, we denote the set of list edges by $E_L$ and the set of quad edges by $E_Q$, so $G = (V, E_L \cup E_Q)$.
In each action executed by node $v$, we perform a consistency check for $v$'s variables, so we can also assume that no inconsistencies appear, like $v \prec v.left$, $v.right \prec v$ or $v$ having multiple quad edges into the same subarea.
We first argue that we get rid of corrupted messages that may exist in an initial state of the system:

\begin{lemma} \label{lemma:convergence:corrupted_messages}
	Given any weakly connected graph $G = (V, E_L \cup E_Q)$ and a set of corrupted messages $M$ spread arbitrarily over all node channels.
	Eventually, $G$ is free of corrupted messages, while staying weakly connected.
\end{lemma}

\begin{proof}
	By definition of \pname{} we do not delete any node but only delegate its reference to \textsc{BuildList} keeping $G$ weakly connected at any point in time.
	Also notice that a corrupted message $m \in M$ cannot be delegated infinitely by the way we defined the \textsc{Linearize} and \textsc{QLinearize} actions.
	Because we assume fair message receipt we know that eventually all messages in $M$ will be processed.
\end{proof}

To show the convergence property, we prove convergence for the sorted list and then show that once the sorted list has stabilized, all desired quad edges will eventually be established.

\begin{lemma} \label{lemma:convergence:list}
	For a weakly connected graph $G = (V, E_L \cup E_Q)$, \pname{} eventually transforms $G$ such that the explicit edges in $E_L$ induce a sorted list w.r.t. $\prec$.
\end{lemma}

\begin{proof}
	In \textsc{QuadTimeout} a node $v$ chooses one of its quad edges $(v,w) \in E_Q$ and delegates it to \textsc{BuildList} (Algorithm~\ref{algo:build_quad_tree}, Line~\ref{algo:build_quad_tree:timeout:delegate}), creating an implicit list edge $(v,w) \in E_L$.
	Since we execute \textsc{QuadTimeout} periodically at each node $v \in V$ and choose quad edges in round-robin fashion, it is guaranteed that eventually each quad edge is delegated to \textsc{BuildList}.
	This implies that the graph $G'=(V,E_L)$ consisting of list edges only eventually becomes weakly connected.
	Thus we can apply Lemma~\ref{lemma:build_list:self_stabilization} to show that the sorted list converges.
\end{proof}

\begin{lemma}[Convergence] \label{lemma:convergence}
	Once the edges in $E_L$ induce a sorted list w.r.t. $\prec$, eventually a legitimate state according to Definition~\ref{def:legal_state} is reached.
\end{lemma}

\begin{proof}
	Condition (a) of Definition~\ref{def:legal_state} is already satisfied, so it remains to show (b).
	Recall that $v$ is able to compute $A(v)$ and the set of subareas $Q(v)$ by locally executing \Call{QuadDivision}{$\{v, v.left, v.right\}$, $P$, $1$}.
	As the sorted list has already converged, $Q(v)$ does not change anymore.
	Let $S \subseteq Q(v)$ be the set of subareas that contain at least one node.
	We show that eventually $v.Q$ contains one node for each of those subareas, i.e., $\forall A \in S\ \exists! w \in v.Q: w \in A$.
	For this we consider an arbitrary subarea $A \in S$ and assume w.l.o.g. that $v \prec w$ for all $w \in A$.
	Note that since nodes $v$ choose subareas $A \in Q(v)$ in round-robin fashion (Algorithm~\ref{algo:build_quad_tree}, Line~\ref{algo:build_quad_tree:timeout:choose_subarea}), it is guaranteed that $v$ chooses $A$ periodically and asks its list neighbor $v.right$ for a node in $A$ as long as $v$ does not have any quad edge to a node in $A$.
	Fix the node $w \in A$ such that $w$ is the outmost left node of $A$ in the ordering $\prec$, i.e., $\forall w' \in A, w' \neq w: w \prec w'$.
	We show that eventually $v$ will receive an implicit edge $(v,w) \in E_Q$ as part of a \textsc{QLinearize} call and thus will add $w$ to $v.Q$, transforming the implicit edge into an explicit one.
	Fix $k \in \mathbb{N}_0$ and assume that there are $k$ nodes lying between $v$ and $w$, i.e., $v \prec v_1 \prec \ldots \prec v_{k+1} \prec w$.
	Observe that any node $v_i$ with $v \prec v_i \prec w$ also needs to have a quad edge to the subarea $A$, since we defined $w$ to be the outmost left node in $A$.
	By definition of our protocol, each node $v_i$ in this chain sends out a  \textsc{QLinearize} request to $v_i.right$, demanding for a node lying within the subarea $A$.
	Thus $w$ receives such a request from $v_k$.
	As $w \in A$, $w$ answers $v_k$ by sending a \textsc{QLinearize} request containing itself back to $v_k$, such that $v_k$ establishes an explicit quad edge $(v_k,w) \in E_Q$.
	Once $v_k$ has established this edge, it answers any incoming \textsc{QLinearize} request coming from $v_{k-1}$ and demanding for a node in $A$ by sending a \textsc{QLinearize} request containing $w$ back to $v_{k-1}$.
	Note that as long as $v_{k-1}$ does not know $w$ yet, $v_k$ receives such \textsc{QLinearize} requests periodically from $v_{k-1}$.
	The chain continues inductively until $v$ has received $w$ from $v_1$, which concludes the proof.
\end{proof}

Now we show the closure property according to Definition~\ref{def:self_stabilization}.
The following lemma follows immediately from Lemma~\ref{lemma:build_list:self_stabilization}:

\begin{lemma} \label{lemma:closure:list}
	If the system is in a legitimate state according to Definition~\ref{def:legal_state}, then the explicit edges in $E_L$ are preserved at any point in time if no nodes join or leave the system.
\end{lemma}

Thus it remains to show closure for quad edges:

\begin{lemma} \label{lemma:closure:quadtree}
	If the system is in a legitimate state according to Definition~\ref{def:legal_state}, then the explicit edges in $E_Q$ are preserved at any point in time if no nodes join or leave the system.
\end{lemma}

\begin{proof}
	Follows from the definition of \pname{}, because once a quad edge is established, we do not remove it anymore.
	Also it is easy to see that any incoming node $w$ that is part of a \textsc{QLinearize} call is just delegated to \textsc{BuildList} by $v$ and not included into $v.Q$.
	This holds, because for the subarea $A \in Q(v)$ which contains $w$ there already has to exist a node $w' \in v.Q$ with $w' \in A$, otherwise this would violate condition (b) of Definition~\ref{def:legal_state}.
\end{proof}

The combination of Lemma~\ref{lemma:closure:list} and Lemma~\ref{lemma:closure:quadtree} implies the following corollary:

\begin{corollary}[Closure] \label{corollary:closure}
	If the system is in a legitimate state according to Definition~\ref{def:legal_state}, then the explicit edges in $E_L \cup E_Q$ are preserved at any point in time if no nodes join or leave the system.
\end{corollary}

Combining Lemma~\ref{lemma:convergence} and Corollary~\ref{corollary:closure} yields the main result of this section:

\begin{theorem}
	\pname{} is self-stabilizing.
\end{theorem}

\subsection{Routing Protocol Analysis} \label{sec:analysis:routing}
In this section we show that \pname{} along with the routing protocol \textsc{Search-Quad} (Algorithm~\ref{algo:search_quad_tree}) satisfies geographic monotonic searchability (Definition~\ref{def:monotonic_searchability:quad}) and thus also monotonic searchability (Definition~\ref{def:monotonic_searchability:classical}).
First we need the following technical lemma stating that for each node $v \in V$ the set $Q(v)$ monotonically increases over time:

\begin{lemma} \label{lemma:monotonic_searchability:prelemma}
	Consider an arbitrary system state at time $t$ and a node $v \in V$.
	Let $Q(v)$ be the output of \textsc{QuadDivision}$(\{v, v.left, v.right\}$, $P$, $1)$ executed at time $t$ and let $Q(v)'$ be the output of \textsc{QuadDivision}$(\{v, v.left, v.right\}$, $P$, $1)$ executed at any point in time $t' > t$.
	Then it holds $Q(v) \subseteq Q(v)'$.
\end{lemma}

\begin{proof}
	By definition of our protocols it holds that if node $v$ locally calls \textsc{QuadDivision}$(\{v, v.left, v.right\}$, $P$, $1)$ in order to compute the set $Q(v)$, then any inconsistencies regarding $v.left$ and $v.right$ are already resolved.
	The lemma then follows from the fact that \textsc{BuildList} does not replace list variables $v.left$ and $v.right$ with nodes that are further away from $v$ than the current entries.
	More formally, consider w.l.o.g. the variable $v.right$ such that $v \prec v.right$.
	By the definition of \textsc{Linearize}, $v$ does not replace $v.right$ by a node $w$ for which $v.right \prec w$ holds.
	This implies that any subsequent \Call{QuadDivision}{$\{v, v.left, v.right\}$, $P$, $1$} call only transfers subareas to $Q(v)$ that are obtained by cutting $A(v)$.
	Therefore, it holds for any subarea $A \in Q(v)$ that $A \in Q(v)'$.
\end{proof}

We are now ready to show the main result of this section:

\begin{theorem} \label{theorem:plane_mon_search}
	\pname{} along with \textsc{SearchQuad} satisfies geographic monotonic searchability.
\end{theorem}

\begin{proof}
	Assume a \Call{Search}{$v$, $(x,y)$} request $S$ terminated and returned $w \in V$ to the initiator $v$ at time $t$, such that $w$ would also be the node that would have been returned if the system already was in a legitimate state.
	Now assume that $v$ initiates another \Call{Search}{$v$, $(x,y)$} request $S'$ at time $t' > t$.
	We show that $S'$ returns $w$ as well.

	Let $(v,v_1,\ldots,v_k,w)$ be the path that has been traversed by $S$.
	We claim that $S'$ traverses the exact same path as $S$.
	Let $Q(v)$ be the output of \Call{QuadDivision}{$\{v, v.left, v.right\}$, $P$, $1$} executed when processing $S$ at $v$ and let $Q(v)'$ be the output of \Call{QuadDivision}{$\{v, v.left, v.right\}$, $P$, $1$} executed when processing $S'$ at $v$.
	Let $A(w) \in Q(v)$ be the subarea that contains $w$ and $A(v_1) \in Q(v)$ be the subarea that contains $v_1$.
	Since $S$ has been delegated by $v$ to $v_1$, it follows from the definition of the \textsc{SearchQuad} protocol (Algorithm~\ref{algo:search_quad_tree}) that $w \in A(v_1)$.
	Lemma~\ref{lemma:monotonic_searchability:prelemma} implies that $Q(v) \subseteq Q'(v)$ and thus $A(v_1) \in Q'(v)$.
	Therefore it follows from the definition of the \textsc{SearchQuad} protocol that $v$ delegates $S'$ to $v_1$ as well.
	By arguing the same way for any node $v_i$ on the remaining path $(v_1,\ldots,v_k,w)$, we can conclude that $S'$ arrives at $w$ and terminates, which finishes the proof.
\end{proof}

As already indicated in Section~\ref{sec:problem_statement}, we obtain the following corollary:

\begin{corollary}
	\pname{} along with \textsc{SearchQuad} satisfies monotonic searchability.
\end{corollary}

Finally, we are able to derive an upper bound on the number of hops for any search message if we assume that the Euclidean distance between any pair $(u,v)\in V$ is at least $\vert\vert (u,v)\vert\vert\geq \frac{1}{n}$.
We start with the following lemma:

\begin{lemma} \label{lemma:search_quad:upper_bound}
	Let $(x,y) \in [0,1]^2$ and suppose a \textsc{Search}$(u$, $(x,y))$ request reached node $v_k$ after $k \in \mathbb{N}_0$ hops, $k$ even.
	Then the maximum Euclidean distance from $v_k$ to the position $(x,y)$ is at most $1/2^{(k-1)/2}$.
\end{lemma}

\begin{proof}
	Let the even number $k \in \mathbb{N}_0$ be the number of hops until \Call{Search}{$u$, $(x,y)$} terminates.
	Let $(u_x, u_y)$ be the coordinates of $u$.
	Initially the Euclidean distance between $(u_x, u_y)$ and $(x,y)$ is maximized, if both coordinates lie on the corners of $P$ such that the straight line between $(u_x, u_y)$ and $(x,y)$ is the diagonal going through $P$.
	
	Note that after two hops we reduce the area in which the target is located by a factor $0.25$.
	When using the Pythagorean theorem to compute the length of the diagonal of the quad, we can compute the maximum distance between the node $v_k$ and $(x,y)$, which is equal to $\sqrt{(1/\sqrt{2^k})^2 + (1/\sqrt{2^k})^2} = 1/\sqrt{2^{k-1}} = 1/2^{(k-1)/2}$.
\end{proof}

We are now ready to prove the following Theorem:

\begin{theorem} \label{theorem:search_quad:hops}
	If for the Euclidean distance between any pair $(u,v) \in V$ it holds $\vert \vert (u,v)\vert \vert \geq 1/n$, then any search message is delegated at most $\mathcal O(\log n)$ times.
\end{theorem}

\begin{proof}
	Assume that a \Call{Search}{$u$, $(x,y)$} message is at node $v_k$ after $k$ hops.
	It is easy to see that after each delegation, the remaining area in which we have to search for $(x,y)$ is halved. 
	We know by Lemma~\ref{lemma:search_quad:upper_bound} that the maximum Euclidean distance from $v_k$ to $(x,y)$ within $k$ hops is at most $1/2^{(k-1)/2}$, when $k$ is even. 
	Set $k=4 \log n$.
	Then the maximum Euclidean distance is at most \[\frac{1}{2^{(4\cdot\log n -1)/2}} = \frac{1}{2^{2\cdot\log n - 1/2}} = \frac{\sqrt{2}}{2^{2\cdot\log n}} = \frac{\sqrt{2}}{n^2} = \mathcal{O}\left(\frac{1}{n^2}\right) < \mathcal{O}\left(\frac{1}{n}\right)\] which implies that the remaining area in which we have to search does not contain a node other than $v_k$, so the routing protocol terminates.
	As $k \in \mathcal O(\log n)$, the theorem follows.
\end{proof}

\section{Generalization to high-dimensional Settings} \label{appendix:multidim}
In this section we discuss how to extend our protocols for high-dimensional settings in order to support self-stabilizing octtrees with geographic monotonic searchability.
Fix a dimension $d > 2$, i.e., we are given a $d$-dimensional hypercube $P$ of unit side length.
Then each node $v$ has coordinates $(v_1,\ldots,v_d) \in [0,1]^d$.

We generalize the \textsc{QuadDivision} procedure as follows: Instead of alternating between two different cuts (vertical and horizontal cuts), we alternate between $d$ different cuts now.
Thus, for all $i \in \{1,\ldots,d\}$ we can define an \emph{$i$-cut} on the (sub-)cube $P$ whose side length in dimension $i$ is equal to $I$ as follows: We assign all points $p \in P$ whose $i^{th}$ coordinate is smaller than $\frac{1}{2} I$ to the subcube $P_1$, and the rest of the points to the subcube $P_2$.
As an example consider Figure~\ref{fig:multi_cutting_example} for a sequence of different cuts on a 3-dimensional hypercube.

\begin{figure}[ht]
 	\centering
 	\includegraphics[width=\textwidth]{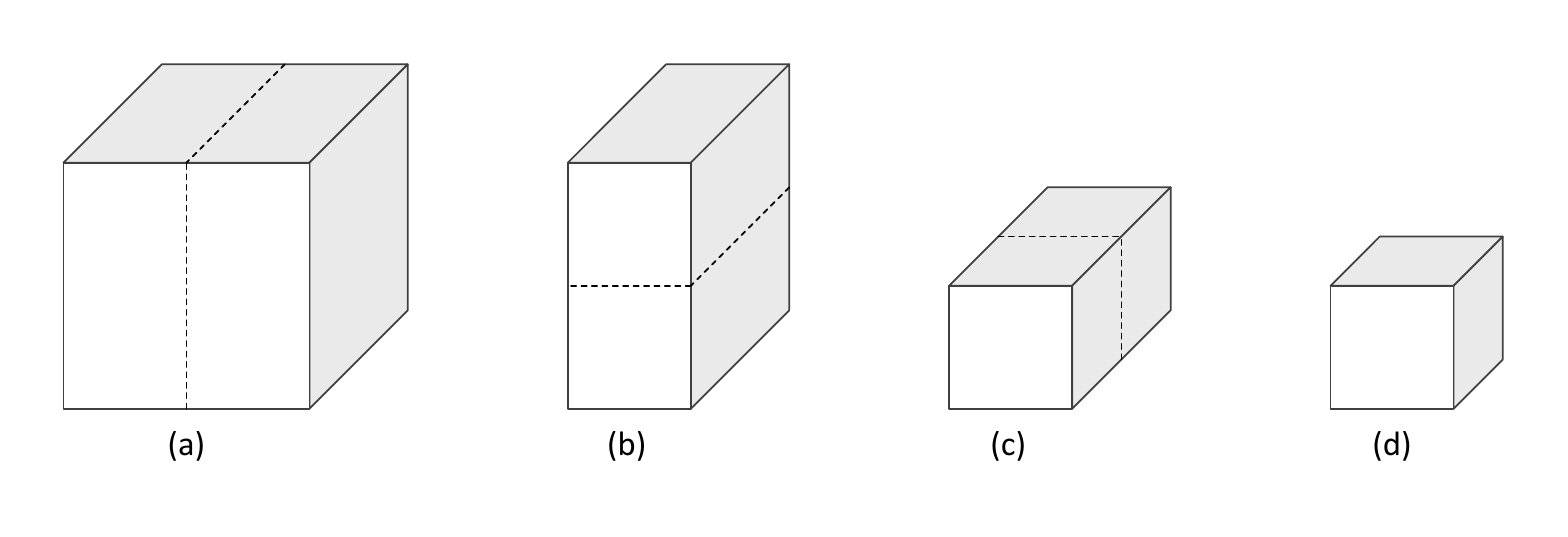}
 	\caption{Illustration for the high-dimensional equivalent of \textsc{QuadDivision}.
 	The sequence shows a $1$-cut ($(a) \rightarrow (b)$), followed by a $2$-cut ($(b) \rightarrow (c)$) and a $3$-cut ($(c) \rightarrow (d)$). The dashed lines indicate how the next cut in the sequence is applied to the (sub-)cube.}
 	\label{fig:multi_cutting_example}
\end{figure}

Thus, the \textsc{QuadDivision} algorithm remains well-defined.

Next, consider the tree $T$ that represents the output of the new \textsc{QuadDivision} algorithm.
$T$ remains a binary tree, however, the levels of the tree now alternate between $d$ different cuts instead of only $2$.
Thus, we obtain the total ordering $\prec$ the same as before, namely by performing a DFS on $T$, always going to the left child first.
This already implies that \textsc{BuildList} is also well-defined in the $d$-dimensional setting.

Last but not least, it is easy to see that one can generalize the definition for $A(v)$ and $Q(v)$ (Definition~\ref{def:target_subareas}) to dimension $d$, since the tree $T$ remains well-defined.
This implies that we have a well-defined legitimate state according to the generalization of Definition~\ref{def:legal_state} and thus the \pname{} protocol along with the routing protocol \textsc{SearchQuad} is well-defined such that all claims made in the main parts of the paper can be generalized to $d$-dimensional settings.

The following corollary summarizes the above discussion:

\begin{corollary}
	There exists a self-stabilizing protocol for a ($d$-dimensional) octtree along with a routing algorithm $\mathcal R$ that satisfies geometrical monotonic searchability.
\end{corollary}

It is also easy to see that the generalized version of \textsc{SearchQuad} delegates a \textsc{Search} message at most $\mathcal O(\log n)$ times until termination in case the Euclidean distance between any two nodes $v, w \in V$ is at least $1/n$.

Finally we want to emphasize that the variables for each node $v \in V$ do not change in higher dimensions for our protocol, making it fairly easy to adapt.

\section{Conclusion and Future Work} \label{sec:conclusion}
In this paper we studied monotonic searchability in high-dimensional settings and came up with a self-stabilizing protocol \textsc{BuildQuadTree} along with its routing protocol \textsc{SearchQuad}.
We showed that \textsc{BuildQuadTree} along with \textsc{SearchQuad} satisfies monotonic searchability, as well as the even stronger variant of geographic monotonic searchability. 
 
For future work, one may consider the dynamic setting, where nodes are able to join or leave the system.
Our protocol can be easily extended to include nodes that join the system at an old node, meaning that an implicit edge is generated.
We then just let \textsc{BuildQuadTree} transform the system to a legitimate state again.
The more interesting scenario is to think of a protocol that allows nodes to leave the system without violating geometric monotonic searchability.
This is a non-trivial task, as a leaving node potentially destroys search paths for other nodes.

%
%
%
\bibliography{literature}{}
\bibliographystyle{plain}

\end{document}